\documentclass[reqno,a4paper]{amsart}
\usepackage{comment,amssymb,latexsym,upref,enumerate,fouridx}
\usepackage{mathrsfs,color}
\usepackage{centernot}

\usepackage[colorlinks,linkcolor=blue,citecolor=blue]{hyperref}
\allowdisplaybreaks
\numberwithin{equation}{section}

\theoremstyle{plain}
\newtheorem{thm}{Theorem}[section]

\newtheorem{prop}[thm]{Proposition}

\theoremstyle{definition}

\theoremstyle{remark}
\newtheorem{rem}[thm]{Remark}

\input Tdef.def

\begin{document}

\title[Future global stability for relativistic perfect fluids with $p=K\rho$, $1/3<K<1/2$]{Future global stability for relativistic perfect fluids with linear equations of state $p=K\rho$ where $1/3<K<1/2$}

\author[T.A. Oliynyk]{Todd A. Oliynyk}
\address{School of Mathematics\\
9 Rainforest Walk\\
Monash University, VIC 3800\\ Australia}
\email{todd.oliynyk@monash.edu}

\begin{abstract}
\noindent 
We establish the future stability of nonlinear perturbations of a class of homogeneous solutions to the relativistic Euler equations with a 
linear equation of state $p=K\rho$ on exponentially expanding FLRW spacetimes for the equation of state parameter values $1/3 < K < 1/2$.
\end{abstract}

\maketitle

\section{Introduction\label{intro}}

Relativistic perfect fluids on a prescribed spacetime $(M,\gt)$ are governed by the relativistic Euler equations given by\footnote{Our indexing conventions are as follows: lower case Latin letters, e.g. $i,j,k$,
will index spacetime coordinate indices that run from $0$ to $3$ while upper case Latin letters, e.g. $I,J,K$, will index spatial coordinate indices that run from
$1$ to $3$.} 
\begin{equation}
\nablat_i \Tt^{ij}=0 \label{relEulA}
\end{equation}
where 
\begin{equation*}
\Tt^{ij} = (\rho+p)\vt^i \vt^j + p \gt^{ij}
\end{equation*}
is the stress energy tensor, $\rho$ is the fluid proper energy density, $p$ is the fluid pressure, and $\vt^{i}$ is the fluid
four-velocity normalized by $\gt_{ij}\vt^i \vb^j=-1$. In this article, we will be interested in analyzing the relativistic
Euler equations
on exponentially expanding Friedmann-Lema\^{i}tre-Robertson-Walker (FLRW) spacetime of the form $(M,\gt)$
where
\begin{equation*}
M = (0,1]\times \Tbb^3
\end{equation*}
and\footnote{By introducing the change of coordinate $\tilde{t}=-\ln(t)$, the metric \eqref{FLRW} can be brought into the more recognizable form
\begin{equation*}
\gt = -d\tilde{t}\otimes d\tilde{t} + e^{2\tilde{t}}\delta_{ij}dx^I \otimes dx^J,
\end{equation*}
where now $\tilde{t} \in [0,\infty)$.
}
\begin{equation} \label{FLRW}
\gt = \frac{1}{t^2} g
\end{equation}
with
\begin{equation} \label{conformal}
g = -dt\otimes dt + \delta_{IJ}dx^I \otimes dx^J.
\end{equation}
It is important to note that, due to our conventions, the future is located in the direction of \textit{decreasing} $t$ and future timelike infinity is located at $t=0$. 
Consequently, we require that
\begin{equation*} 
\vt^0 < 0
\end{equation*}
in order to ensure that the four-velocity is future directed.

The future stability of nonlinear perturbations of homogeneous solutions to the relativistic Euler equations with a linear equation of 
state 
\begin{equation*} 
p = K \rho
\end{equation*}
on exponentially expanding FLRW spacetimes has been well studied for the 
parameter range
\begin{equation*} 
0\leq K\leq \frac{1}{3}.
\end{equation*}
The first such stability result
was, building
on the earlier stability results for the Einstein-scalar field system \cite{Ringstrom:2008}, established\footnote{In these articles, stability was established in the more difficult case where the fluid is coupled to the Einstein equation. However, the techniques used there also work in the simpler setting considered in this article where coupling to gravity is neglected.}
for the parameter values $0<K<1/3$ in the articles \cite{RodnianskiSpeck:2013,Speck:2012}. Stability results for the end points $K=1/3$ and $K=0$ were established later in \cite{LubbeKroon:2013} and \cite{HadzicSpeck:2015},
respectively. See also \cite{Friedrich:2017,LiuOliynyk:2018b,LiuOliynyk:2018a,Oliynyk:CMP_2016} for different proofs and perspectives, the articles \cite{LeFlochWei:2021,LiuWei:2021} for related stability results for fluids with nonlinear equations of state, and the articles \cite{FOW:2021,Ringstrom:2009,Speck:2013,Wei:2018} for stability results on other expanding FLRW spacetimes (e.g. power law expansion). The importance of all of 
these works is that they demonstrate spacetime expansion can suppress shock formation in fluids, which was first discovered in the Newtonian cosmological setting \cite{BrauerRendallReula:1994}. This should be compared to the work of \cite{Christodoulou:2007} where it is established that arbitrary small perturbations of a class of homogeneous solutions to the relativistic Euler equations, for relatively general equations of state, on Minkowski spacetime, which is a FLRW spacetime with spatial manifold $\Rbb^3$ and no expansion, form shocks in finite time.

For linear equations of states, the parameter $K$ determines the square of the sound speed, and consequently, it is natural to assume\footnote{While this restriction on the sound speed is often taken for granted and implicitly assumed, it is strictly speaking not necessary; see \cite{Geroch:2010} for an extended discussion.} 
that $K$ satisfies
\begin{equation} \label{KrangeC}
0\leq K \leq 1
\end{equation}
so that the propagation speed for the fluid is less than or equal to the speed of light.
When the sound speed is equal to the speed of light, that is $K=1$, it is well known that the irrotational relativistic Euler equations coincide, under a change of variables, with the linear wave equation. In this case, the future global existence of solutions on exponentially expanding FLRW spacetimes can be inferred from standard existence results for linear wave equations. This leaves us to consider the
parameter range
\begin{equation} \label{KrangeA}
\frac{1}{3}<K<1,
\end{equation}
which we will assume holds for the remainder of the article.

The asymptotic behavior of relativistic fluids on exponentially expanding FLRW spacetimes with a linear equation of state for $K$ satisfying \eqref{KrangeC} was investigated in the article \cite{Rendall:2004} by Rendall using formal expansions. In that article, Rendall observed that the formal expansions can become inconsistent for $K$ in the range \eqref{KrangeA} if the leading order term in the expansion of the four-velocity vanishes somewhere. In that case, he speculated that inconsistent behavior in the expansions could be due to inhomogeneous
features developing in the fluid density that would lead to the density contrast blowing up.  This possibility for instability in solutions to the relativistic Euler equations for the parameter range \eqref{KrangeA} was also commented on by Speck in \cite[\S 1.2.3]{Speck:2013}. There, Speck presents a heuristic analysis that suggest uninhibited growth should set in for solutions of the relativistic Euler equations for the parameter values
\eqref{KrangeA}. These speculations leave the existence of future global solutions to the relativistic Euler equations in doubt for $K$
satisfying \eqref{KrangeA}.

In this article, we rule out, under a small initial data hypothesis, the possibility of any pathologies developing in finite time for $K$ satisfying
\begin{equation*}
\frac{1}{3}<K < \frac{1}{2}
\end{equation*}
by establishing, for these parameter values, the future stability of nonlinear perturbations of a  class of homogeneous solutions, see \eqref{Homsol}, to the relativistic Euler equations on exponentially expanding FLRW spacetimes. For a precise statement of our stability result, see Theorem \ref{mainthm}, which is the main result of this article. This, of course, leaves open the possibility of finite time blow-up for $K$ satisfying $1/2 < K < 1$. As a first step towards understanding the behavior of solutions
in this regime, we establish in Theorem \ref{symthm} the future stability of $\Tbb^2$-symmetric nonlinear perturbations of the same class of homogenous solutions for the full parameter range $1/3 < K < 1$.  Here, the stability proof relies heavily on the $\Tbb^2$ symmetry that allows us to reduce the relativistic Euler equations to an essentially regular $1+1$ dimensional problem. It is unclear at the moment if one should
expect that this result will still hold for $K$ satisfying $1/2 \leq K < 1$ if the $\Tbb^2$-symmetry assumption is removed. We plan to revisit this interesting question in a separate article.

\subsection{Overview}
The proof of our main stability result, Theorem \ref{mainthm}, is based on the Fuschsian method for establishing the global existence of solutions to systems of hyperbolic equations that was first employed in \cite{Oliynyk:CMP_2016} and further developed in the articles \cite{BOOS:2020,FOW:2021,LiuOliynyk:2018b,LiuOliynyk:2018a}. This method 
relies on transforming the global existence problem for a given hyperbolic system into an existence problem for a
Fuchsian symmetric hyperbolic
equation of the form
\begin{equation*}
\Asc^0(t,\Wsc)\del{t}\Wsc+ \Asc^i(t,\Wsc)\del{i}\Wsc  = \frac{1}{t}\Af(t,\Wsc)\Pbb \Wsc + \Fsc(t,\Wsc)  
\end{equation*}
on a finite time interval $(0,T_0]$.
Once in this form, the existence of solutions on the time interval $(0,T_0]$ can be deduced, under a suitable smallness assumption
on the initial data specified at time $t=T_0$, from general existence theorems for such Fuchsian systems that have been established in
the articles  \cite{BOOS:2020,FOW:2021,LiuOliynyk:2018b,LiuOliynyk:2018a}. 

One of the main advantages of the Fuchsian method is that it brings into clear focus the structure in evolution
equations that ensures stability. For the relativistic Euler equations with $K>1/3$ on exponential expanding
FLRW spacetimes, the stability structure is particularly well hidden in the standard formulation. One of the main reasons for this is that the homogeneous solutions about which the solutions are perturbed are more
complicated for $K>1/3$ compared to when $K\leq 1/3$. As shown in Proposition \ref{Homprop}, the homogeneous solutions are of the form
\begin{equation*}
\bigl(\rho,\vt^i) = \biggl( \frac{\rho_c t^{\frac{2(1+K)}{1-K}}}{(t^{2\mu}+ e^{2u})^{\frac{1+K}{2}}},
-t^{1-\mu}\sqrt{e^{2u}+t^{2 \mu} }, t^{1-\mu }e^{u},0,0\biggr), \quad 0<t\leq 1,
\end{equation*}
where $\rho_c>0$ is a constant and $u(t)$ is a function satisfying\footnote{Note that $u$ is determined by solving
the IVP \eqref{HomeqB.1}-\eqref{HomeqB.2}.}
\begin{equation*}
u(t)=u(0)+\Ord(t^{2\mu}) \AND u'(t)=\Ord(t^{2\mu-1}).
\end{equation*}
Here, $\mu$ is a constant defined by
\begin{equation*}
\mu = \frac{3K-1}{1-K}
\end{equation*}
and $u(0)$ is the asymptotic value of $u$ at future time-like infinity. The fact that the spatial three-velocity 
$(\vt^I)=(t^{1-\mu }e^{u},0,0)$ of the homogeneous solution is no longer trivial, unlike for $K\leq 1/3$ where it is trivial, ultimately
results in a significant difference in the behavior under non-linear perturbations for the component $\vt^1$ versus the
components $\vt^2$ and $\vt^3$. This more complicated behavior is primarily responsible for the  increased difficulty
in establishing stability for $K>1/3$ compared to $K\leq 1/3$, and for obscuring the structure in the Euler equations that allows for
global existence, which
now requires a much more involved choice of variables to make apparent.

In this article, we transform the relativistic Euler equations into a suitable Fuchsian form in a number of steps. We start in Section
 \ref{relEulsec} with a formulation, see \eqref{relEulB}, of the relativistic Euler equations that was first employed in \cite{Oliynyk:CMP_2015}. We then modify this system by introducing a new density variable defined by \eqref{zetatdef}, which results in the system \eqref{relEulC}. In Section
 \ref{Ftrans},  we proceed by decomposing the conformal three-velocity $v_I$ into its length determined by the variables $u$, $w_1$ and into
 a normalized vector determined by the variables $w_2$, $w_3$; see  \eqref{cov2a}-\eqref{cov2c} for the relevant formulas. Here, $u$ only depends on $t$ and is used to parameterize a class of homogeneous solutions of relativistic Euler equations. After some straightforward, but lengthy
 calculations, we obtain two equivalent versions of Euler equations, now expressed in terms of the new variables $u$, $w_1$, $w_2$ and $w_3$,  given by \eqref{relEulD} and \eqref{relEulE}. We then use the second version \eqref{relEulE}, see Section \ref{Homsec}, to identify the ODE satisfied by $u$ that determines homogeneous solutions of the relativistic Euler equations. The existence of solutions to
 this ODE is established in Proposition \ref{Homprop}.

The transformation of the relativistic Euler equations into a suitable Fuchsian form
is then completed in Section \ref{Fform}, which results in the Fuchsian equation \eqref{relEulK}. The coefficients of this Fuchsian equation are analyzed in
Section \ref{coeff} in order to verify that this system satisfies the required properties in order to apply the existence theory
from \cite{BOOS:2020}. It is worth noting that the restriction $1/3<K<1/2$ occurs at this step. The type
of singular terms that appear, see Remark \ref{dtA0rem} below,  in the coefficients of \eqref{relEulK} are of the form $t^{2\mu-1}$ and $t^{\mu-1}$ for
$0\leq \mu \leq 1/2$, and $t^{1-2\mu}$ and $t^{-\mu}$ for $1/2<\mu$. In order to apply the existence theory from
\cite{BOOS:2020}, we need $\mu$ to satisfy $0< \mu <1$, which by the definition above, corresponds to
$1/3<K<1/2$.

The existence theory from \cite{BOOS:2020} is then applied to \eqref{relEulK} in the proof of Theorem \ref{mainthm}, which can be found
in Section \ref{stability}, to establish the future stability of nonlinear perturbations of the homogeneous solutions to the relativistic Euler
equations from Proposition \ref{Homprop}. Finally, in Section \ref{symsec}, we establish the future stability of $\Tbb^2$-symmetric nonlinear perturbations of the same class of homogenous solutions for the full parameter range $1/3 < K < 1$
by using the $\Tbb^2$ symmetry to reduce the relativistic Euler equations to an essentially regular $1+1$ dimensional problem. The precise statement of the stability result in this setting is given in Theorem \ref{symthm}.

\subsection{Outlook and future directions}
The most natural and physically relevant generalization of the results of this article would be to establish an analogous stability result for the coupled Einstein-Euler equations with a positive cosmological constant $\Lambda >0$ and  for $1/3<K<1/2$. While there are details to check, we expect that this result will follow from
a straightforward adaptation of the arguments from \cite{Oliynyk:CMP_2016}. The evidence for this expectation comes from the behavior of the term\footnote{Here, $v_i$ is the conformal fluid four-velocity defined by \eqref{cov1}.}  $t^{-2}\rho v_i v_j$, which is really the only possible problematic term that could,
if it grew too quickly as $t\searrow 0$, prevent the use of the gravitation variables used in \cite{Oliynyk:CMP_2016}
to bring the Einstein equations into Fuchsian form. However, by Theorem \ref{mainthm}, we know that
\begin{equation*}
\frac{1}{t^2}\rho v_i v_j = \Ord(t^2),
\end{equation*}
and so, we see that $t^{-2}\rho v_i v_j$ is well behaved as $t\searrow 0$. 

It for this reason that we have not considered the coupled
Einstein-Euler equations here, and instead, we have focused our attention on the relativistic Euler equations on an exponentially expanding FLRW background, which we believe
is advantageous because it not only simplifies the presentation, but also addresses all the essential technical difficulties.
With that said, we do plan to establish an analogous global existence result for the Einstein-Euler equations in a separate article.

\section{A symmetric hyperbolic formulation of the Relativistic Euler equations\label{relEulsec}}
The first step in transforming the relativistic Euler equations \eqref{relEulA} into a suitable Fuchsian form is to find a symmetric hyperbolic formulation of the relativistic Euler equations.
Here, we start with the symmetric hyperbolic formulation derived in \cite[\S 2.2]{Oliynyk:CMP_2015}, see also \cite[\S 2.2]{Oliynyk:CMP_2016}.
This involves introducing the \textit{conformal fluid four-velocity} $v_i$ and the \text{modified density} $\zeta$ according to
\begin{equation} \label{cov1}
v_i= \frac{1}{t} g_{ij}\vt^{j} \AND \rho = t^{3(1+K)}\rho_c e^{(1+K)\zeta}, \quad \rho_c\in \Rbb_{>0}.
\end{equation} 
Using these variables, the computations carried out in  \cite[\S 2.2]{Oliynyk:CMP_2015} show that relativistic Euler equations can be
cast into the following symmetric hyperbolic form:
\begin{equation} \label{relEulB}
B^k \del{k}V = \frac{1}{t}\Bc \pi V
\end{equation}
where
\begin{align}
V &= (\zeta, v_J )^{\tr} ,  \label{Vdef}\\
v_0 & = \sqrt{|v|^2 +1} , \qquad |v|^2 = \delta^{IJ}v_I v_J,  \label{v0def}\\
v^i & = \delta^{iJ}v_J - \delta^{i}_0 v_0,  \label{viupdef}\\
\Bc &= \frac{-1}{v^0}\begin{pmatrix} 1 & 0 \\ 0 & \frac{1-3K}{v_0}\delta^{JI} \end{pmatrix}, \label{Bcdef}\\
\pi &= \begin{pmatrix} 0 & 0 \\ 0 & \delta_{I}^J \end{pmatrix}, \\
L^k_I &= \delta^k_J - \frac{v_J}{v_0} \delta^k_0, \\
M_{IJ} &= \delta_{IJ} - \frac{1}{(v_0)^2}v_I v_J,  \label{Mdef}\\
B^0 &= \begin{pmatrix} K  &  \frac{K}{v^0}  L^0_M \delta^{MJ} \\ \frac{K}{v^0} \delta^{LI} L^0_L & \delta^{LI} M_{LM} \delta^{MJ} \end{pmatrix}
\intertext{and}
B^K &= \frac{1}{v^0}\begin{pmatrix} Kv^K  &  K  L^K_M \delta^{MJ} \\ K \delta^{LI} L^K_L & \delta^{LI} M_{LM} \delta^{MJ} v^K \end{pmatrix}.
\label{BKdef}
\end{align}

Defining a new modified density variable $\zetat$ by
\begin{equation} \label{zetatdef}
\zetat = \zeta + \ln(v_0),
\end{equation}
we obtain from differentiating \eqref{Vdef} the relation
\begin{equation} \label{dV2dVt}
\del{k} V = Q\del{k}\Vt
\end{equation}
where
\begin{equation} \label{Vtdef}
\Vt = (\zetat, v_J)^{\tr}
\end{equation}
and
\begin{equation*} 
Q= \begin{pmatrix} 1 &  -\frac{1}{v_0^2}v^J \\
0 & \delta^J_I \end{pmatrix}.
\end{equation*}
From \eqref{Vdef}-\eqref{BKdef} and \eqref{dV2dVt}, we see, after multiplying \eqref{relEulB} on the left by $Q^{\tr}$, that the relativistic Euler equations can be expressed 
in terms of the new variables \eqref{Vtdef} as
\begin{equation} \label{relEulC}
\Bt^k \del{k}\Vt = \frac{1}{t}\Bc \pi \Vt
\end{equation}
where 
\begin{align}
\Bt^0 &:= Q^{\tr}B^0 Q = \begin{pmatrix} K & 0 \\ 0 & \delta^{IL}M_{LM}\delta^{MJ}-\frac{K}{v_0^4}v^I v^J\end{pmatrix} 
\label{Bt0def}
\intertext{and}
\Bt^K &:= Q^{\tr}B^K Q =  -\frac{1}{v_0}\begin{pmatrix} K v^K & -\frac{K}{v_0^2}v^J v^K + K\delta^{KJ} \\ 
  -\frac{K}{v_0^2}v^I v^K + K\delta^{KI}&  \bigl(\delta^{IL}M_{LM}\delta^{MJ}+\frac{K}{v_0^4}v^I v^J\bigr)v^K -\frac{K}{v_0^2}
  (v^I \delta^{JK}+v^J \delta^{IK}) \end{pmatrix}. \label{BtIdef}
\end{align}

\section{Transformation to Fuchsian form\label{Ftrans}}
We proceed with the transformation of the relativistic Euler equations into a suitable Fuchsian form
by defining a second change of variables via
\begin{align} 
v_1 &= \frac{t^{-\mu } e^{u(t)+w_1}}{\sqrt{t^{2 \mu
   }
   \left((w_2-w_3)^2+(w_2+w_3)^2\right)+1}}, \label{cov2a} \\
v_2 &= \frac{(w_2+w_3) e^{u(t)+w_1}}{\sqrt{t^{2
   \mu }
   \left((w_2-w_3)^2+(w_2+w_3)^2\right)+1}} \label{cov2b}
 \intertext{and}
v_3 &= \frac{(w_2-w_3) e^{u(t)+w_1}}{\sqrt{t^{2
   \mu }
   \left((w_2-w_3)^2+(w_2+w_3)^2\right)+1}},   \label{cov2c}
\end{align}
where $u(t)$ is a time dependent function  and $\mu\in \Rbb$ is a constant both of which will be fixed below.
Using \eqref{cov2a}-\eqref{cov2c}, we find from differentiating \eqref{Vtdef} that
\begin{align}
\del{t} \Vt &= P \del{t}W + Z  \label{dVt2dWa}
\intertext{and}
\del{I}\Vt &= P \del{I} W \label{dVt2dWb}
\end{align}
where
\begin{gather}
W = (\zetat, w_1, w_2, w_3 )^{\tr} , \label{Wdef}\\ 
Z  = \begin{pmatrix}
 0 \\
 \frac{t^{-\mu -1} e^{\wbr_1} \left(t u'(t)
   \phi-\mu 
   \left(4 t^{2 \mu }
   \left(w_2^2+w_3^2\right)+1\right)\right)}{\phi^{3/2}} \\
 \frac{(w_2+w_3) e^{\wbr_1} \left(t
   u'(t) \phi-2 \mu 
   t^{2 \mu }
   \left(w_2^2+w_3^2\right)\right)}{t \phi^{3/2}} \\
 \frac{(w_2-w_3) e^{\wbr_1} \left(t
   u'(t) \phi -2 \mu 
   t^{2 \mu }
   \left(w_2^2+w_3^2\right)\right)}{t \phi^{3/2}} 
\end{pmatrix}, \label{Zdef} \\
P=\begin{pmatrix}
 1 & 0 & 0 & 0 \\
 0 & \frac{t^{-\mu } e^{\wbr_1}}{\sqrt{\phi}} &
   -\frac{2 w_2 t^{\mu }
   e^{\wbr_1}}{\phi^{3/2}} &
   -\frac{2 w_3 t^{\mu }
   e^{\wbr_1}}{\phi^{3/2}} \\
 0 & \frac{(w_2+w_3)
   e^{\wbr_1}}{\sqrt{\phi}} &
   -\frac{e^{\wbr_1}\eta_3}{\phi^{3/2}} &
   \frac{e^{\wbr_1} \eta_2}{\phi^{3/2}} \\
 0 & \frac{(w_2-w_3)
   e^{\wbr_1}}{\sqrt{\phi}} &
   \frac{e^{\wbr_1}\xi_3}{\phi^{3/2}} &
   -\frac{e^{\wbr_1} \xi_2}{\phi^{3/2}} 
\end{pmatrix} \label{Pdef}
\end{gather}
and we have set
\begin{align}
\wbr_1 &= u+w_1, \label{wbr1def} \\
\phi &= 2 t^{2\mu } \left(w_2^2+w_3^2\right)+1, \label{phidef} \\
\eta_\Lambda &= \left(2 w_\Lambda t^{2 \mu } 
   (w_2-w_3)+(-1)^\Lambda 1\right),  \quad \Lambda =2,3, \label{etadef}
\intertext{and}
\xi_\Lambda &= \left(2 w_\Lambda t^{2 \mu }
   (w_2+w_3)+1\right), \quad \Lambda =2,3. \label{xidef}
\end{align}
By multiplying \eqref{relEulC} on the left by $P^{\tr}$, we see, with the help of \eqref{dVt2dWa}-\eqref{dVt2dWb}, that $W$ satisfies
\begin{equation} \label{relEulD}
A^0 \del{t}W + A^I \del{I}W = \frac{1}{t}P^{\tr}\bigl(\Bc \pi \Vt-t \Bt^0 Z)
\end{equation}
where
\begin{equation} \label{Aidef}
A^i = P^{\tr} \Bt^i P.
\end{equation}

Next, setting
\begin{equation} \label{Acdef}
\Ac^I = (A^0)^{-1}A^I
\end{equation}
and 
\begin{equation} \label{Fcdef}
\Fc =  \frac{1}{t}(A^0)^{-1} P^{\tr}\bigl(\Bc \pi \Vt-t \Bt^0 Z),
\end{equation}
we deduce from \eqref{relEulD} that $W$ satisfies
\begin{equation} \label{relEulE}
\del{t}W + \Ac^I \del{I}W = \Fc.
\end{equation}
Moreover, straightforward, but lengthy, calculations using \eqref{v0def}-\eqref{Mdef},  \eqref{Vtdef}, \eqref{Bt0def}-\eqref{BtIdef}, \eqref{cov2a}-\eqref{cov2c}, \eqref{Zdef}-\eqref{xidef} and \eqref{Aidef}-\eqref{Fcdef} yield the following
explicit formulas for the matrices
$A^0$, $\Ac^I$ and the source term $\Fc$:
\begin{equation}  \label{A0rep}
A^0 = \begin{pmatrix}
 K & 0 & 0 & 0 \\
 0 & \frac{t^{2 \mu } e^{2 \wbr_1}-(K-1) e^{4
   \wbr_1}}{\psi^2} & 0 & 0 \\
 0 & 0 & \frac{2 e^{2 \wbr_1} \left(2 w_3^2
   t^{2 \mu }+1\right)}{\phi^2} &
   -\frac{4 w_2 w_3 t^{2 \mu } e^{2
   \wbr_1}}{\phi^2} \\
 0 & 0 & -\frac{4 w_2 w_3 t^{2 \mu } e^{2
   \wbr_1}}{\phi^2} &
   \frac{2 e^{2 \wbr_1} \left(2 w_2^2 t^{2
   \mu }+1\right)}{\phi^2} \\
\end{pmatrix},
\end{equation}
\begin{equation} \label{Ac1rep}
\Ac^1 = \frac{1}{\sqrt{\frac{t^{2\mu}}{e^{2\wt_1}}+1}}\begin{pmatrix}
 -\frac{1}{\sqrt{\phi}} &
   -\frac{t^{2 \mu }}{\psi \sqrt{\phi}} &
   \frac{2 t^{2\mu } w_2}{\phi^{3/2}} &
   \frac{t^{2\mu} w_3}{\phi^{3/2}}
   \\
 -\frac{K t^{2 \mu } e^{-2
   \wbr_1} \psi}{\sqrt{\phi}\chi} &
   \frac{(2 K-1)
   t^{2 \mu }+(K-1) e^{2
   \wbr_1}}{\sqrt{\phi} \chi} & -\frac{2 K
   t^{2\mu } \psi w_2}{\phi^{3/2}\chi} & -\frac{2 K t^{2\mu }
   \psi w_3}{\phi^{3/2}\chi} \\
 K t^{2\mu }w_2 e^{-2 \wbr_1}
   \sqrt{\phi} & -\frac{K
   t^{2\mu } w_2 \sqrt{\phi}}{\psi} &
   -\frac{1}{\sqrt{\phi}} & 0
   \\
  K t^{2\mu }w_3 e^{-2 \wbr_1} \sqrt{\phi} &
   -\frac{K t^{2\mu } w3 \sqrt{\phi}}{ \psi} & 0 &
  - \frac{1}{\sqrt{\phi}} \\
\end{pmatrix},
\end{equation}
\begin{equation}  \label{Ac2rep}
\Ac^2 = \frac{1}{\sqrt{\frac{t^{2\mu}}{e^{2\wt_1}}+1}}\begin{pmatrix}
 -\frac{t^{\mu } (w_3+w_2)}{\sqrt{\phi}} &
   -\frac{t^{3 \mu } (w_3+w_2)}{\psi \sqrt{\phi}} &
   \frac{t^{\mu } \eta_3}{\phi^{3/2}} &
   -\frac{t^\mu \eta_2}{\phi^{3/2}}
   \\
 -\frac{K t^{3 \mu } (w_2+w_3) e^{-2
   \wbr_1} \psi}{\sqrt{\phi}\chi} &
   \frac{t^{\mu } (w_2+w_3) \left((2 K-1)
   t^{2 \mu }+(K-1) e^{2
   \wbr_1}\right)}{\sqrt{\phi} \chi} & -\frac{K
   t^{\mu } \psi \eta_3}{\phi^{3/2}\chi} & \frac{K t^{\mu }
   \psi\eta_2}{\phi^{3/2}\chi} \\
 -\frac{1}{2} K t^{\mu } e^{-2 \wbr_1}
   \sqrt{\phi} & \frac{K
   t^{\mu } \sqrt{\phi}}{2 \psi} &
   -\frac{t^{\mu } (w_3+w_2)}{\sqrt{\phi}} & 0
   \\
 -\frac{1}{2} K t^{\mu } e^{-2 \wbr_1} \sqrt{\phi} &
   \frac{K t^{\mu } \sqrt{\phi}}{2 \psi} & 0 &
   -\frac{t^{\mu } (w_3+w_2)}{\sqrt{\phi}} \\
\end{pmatrix},
\end{equation}
\begin{equation}  \label{Ac3rep}
\Ac^3 = \frac{1}{\sqrt{\frac{t^{2\mu}}{e^{2\wt_1}}+1}}\begin{pmatrix}
 \frac{t^{\mu } (w_3-w_2)}{\sqrt{\phi}} &
   \frac{t^{3 \mu } (w_3-w_2)}{\psi \sqrt{\phi}} &
   -\frac{t^{\mu } \xi_3}{\phi^{3/2}} &
   \frac{t^\mu \xi_2}{\phi^{3/2}}
   \\
 -\frac{K t^{3 \mu } (w_2-w_3) e^{-2
   \wbr_1} \psi}{\sqrt{\phi}\chi} &
   \frac{t^{\mu } (w_2-w_3) \left((2 K-1)
   t^{2 \mu }+(K-1) e^{2
   \wbr_1}\right)}{\sqrt{\phi} \chi} & \frac{K
   t^{\mu } \psi \xi_3}{\phi^{3/2}\chi} & -\frac{K t^{\mu }
   \psi\xi_2}{\phi^{3/2}\chi} \\
 -\frac{1}{2} K t^{\mu } e^{-2 \wbr_1}
   \sqrt{\phi} & \frac{K
   t^{\mu } \sqrt{\phi}}{2 \psi} &
   \frac{t^{\mu } (w_3-w_2)}{\sqrt{\phi}} & 0
   \\
 \frac{1}{2} K t^{\mu } e^{-2 \wbr_1} \sqrt{\phi} &
   -\frac{K t^{\mu } \sqrt{\phi}}{2 \psi} & 0 &
   \frac{t^{\mu } (w_3-w_2)}{\sqrt{\phi}} \\
\end{pmatrix}
\end{equation}
and
\begin{equation} \label{Fcrep}
\Fc = t^{2\mu -1}\Gc + \Fc_0,
\end{equation}
where
\begin{gather}
\psi = t^{2 \mu }+e^{2 \wbr_1}, \label{psidef}\\
\chi = t^{2\mu }-(K-1) e^{2 \wbr_1}, \label{chidef}
\end{gather}
\begin{equation} \label{Gcdef}
\Gc =\begin{pmatrix}
 0 \\
 -\frac{K (3 K-1) \left(e^{2 w_1}-1\right) e^{2
   u}}{\left((K-1) e^{2 u}-t^{2 \mu }\right)
   \left((K-1) e^{2 \wbr_1}-t^{2 \mu
   }\right)} \\
 0 \\
 0 
\end{pmatrix},
\end{equation}
\begin{equation}  \label{Fc0def}
\Fc_0 = -\frac{\mu}{t}\Pi W
+ \frac{1}{t}
\begin{pmatrix}
 0 \\
 \frac{-(\mu -3 K+1) t^{2 \mu }+t u'\!(t) \left(t^{2
   \mu }-(K-1) e^{2 u(t)}\right)+(-\mu +(\mu
   +3) K-1) e^{2 u(t)}}{ \left((K-1) e^{2 u(t)}-t^{2
   \mu }\right)} \\
 0 \\
 0 
\end{pmatrix}
\end{equation}
and
\begin{equation} \label{Pidef}
\Pi = \begin{pmatrix} 0 & 0 & 0 & 0 \\ 0 & 0 & 0 & 0 \\ 0 & 0 & 1 & 0 \\ 0 & 0 & 0 & 1 \end{pmatrix}.
\end{equation}
For later use, we also define
\begin{equation} \label{Piperpdef}
\Pi^\perp = \id - \Pi,
\end{equation}
and observe that $\Pi$ and $\Pi^\perp$ satisfy the relations
\begin{equation} 
\Pi^2 = \Pi, \quad (\Pi^\perp)^2 = \Pi^\perp, \quad \Pi\Pi^\perp =\Pi^\perp \Pi = 0 \AND \Pi+\Pi^\perp = \id. \label{Pirel} 
\end{equation}

\subsection{Homogeneous solutions\label{Homsec}}
To proceed, we need to identify the homogeneous solutions that we will show are stable to the future under nonlinear perturbations. We
locate these solutions by noting from \eqref{Gcdef} that
\begin{equation*}
\Gc|_{W=0} = 0.
\end{equation*}
From this, \eqref{Fcrep} and \eqref{Fc0def}, it is then clear that the trivial solution $W=0$ will solve \eqref{relEulE} provided that $\mu$ and $u(t)$ are chosen to satisfy
\begin{equation} \label{mufixA}
-\mu +(\mu
   +3) K-1 = 0
\end{equation}
and
\begin{equation} \label{HomeqA}
 u'\!(t) =\frac{(\mu -3 K+1) t^{2 \mu-1}}{t^{2
   \mu }-(K-1) e^{2 u(t)}},
\end{equation}
respectively.
Solving \eqref{mufixA} for $\mu$ yields
\begin{equation} \label{mufixB}
\mu = \frac{3K-1}{1-K},
\end{equation}
which we observe by \eqref{KrangeA} satisfies
\begin{equation} \label{mufixC}
\mu > 0.
\end{equation}
Moreover using \eqref{mufixB}, we note that \eqref{HomeqA} can be expressed as
 \begin{equation} \label{HomeqB}
 u'\!(t) =\frac{K\mu t^{2 \mu-1}}{t^{2
   \mu }+(1-K) e^{2 u(t)}}.
\end{equation}
The following proposition guarantees the existence of solutions to \eqref{HomeqB} that exist for all $t\in (0,1]$.

\begin{prop} \label{Homprop}
Suppose $1/3<K<1$, $\mu = (3K-1)/(1-K)$, and $u_0 \in \Rbb$. Then there exists a unique solution
$u \in C^\infty((0,1]) \cap C^0([0,1])$
to the initial value problem
\begin{align}
 u'\!(t) &=\frac{K\mu t^{2 \mu-1}}{t^{2
   \mu }+(1-K) e^{2 u(t)}},\quad 0<t\leq 1, \label{HomeqB.1} \\
   u(1) &= u_0, \label{HomeqB.2}
\end{align}
that satisfies
\begin{equation}
|u(t)-u(0)| \lesssim t^{2\mu} \AND |u'\!(t)| \lesssim t^{2\mu-1} \label{Hombounds}
\end{equation}
for all $t\in (0,1]$. Moreover, for each $\rho_c\in \Rbb>0$, the solution $u$ determines a homogenous solution of the relativistic Euler \eqref{relEulA}
equations given by
\begin{equation} \label{Homsol}
\bigl(\rho,\vt^i) = \biggl( \frac{\rho_c t^{\frac{2(1+K)}{1-K}}}{(t^{2\mu}+ e^{2u})^{\frac{1+K}{2}}},
-t^{1-\mu}\sqrt{e^{2u}+t^{2 \mu} }, t^{1-\mu }e^{u},0,0\biggr) .
\end{equation}
\end{prop}
\begin{proof}
By standard local existence theorems for ODEs, we know there exists a $T \in [0,1)$ and a unique solution 
$u \in C^\infty((T,1])$ to the initial value problem \eqref{HomeqB.1}-\eqref{HomeqB.2} that can be continued
to smaller times as long as $u(t)$ stays bounded. 
On the other hand, we observe that 
\eqref{HomeqB.1} can be integrated directly to yield the implicit solution
\begin{equation*}
{\textstyle \frac{K}{2}} \ln \bigl(t^{2 \mu }+e^{2
   u(t)}\bigr)- u(t)=c,
\end{equation*}
where the constant $c$ is uniquely determined by the initial condition $u_0$ and the constants $K,\mu$. Solving for $t^{2\mu}$ shows that
\begin{equation*}
e^{\frac{2(c+u(t))}{K}}-e^{2 u(t)} = t^{2\mu}.
\end{equation*}
Since $\mu>0$, this implies the inequality
\begin{equation*}
0\leq e^{\frac{2(c+u(t))}{K}}-e^{2 u(t)} \leq 1, \quad 0\leq T<t \leq 1,
\end{equation*}
from which we deduce the lower bound
\begin{equation} \label{Homprop1}
e^{2 u(t)} \leq e^{\frac{2(c+u(t))}{K}}\quad \Longrightarrow  \quad 2 u(t) \leq  \frac{2(c+u(t))}{K} \quad \Longrightarrow\quad -\frac{c}{1-K}\leq u(t), \quad 0\leq T<t \leq 1.
\end{equation}
On the other hand, since the right hand side of the ODE \eqref{HomeqB.1} is  positive, $u(t)$ must be increasing, and hence, it is bounded above by
\begin{equation} \label{Homprop2}
u(t) \leq u_0  \quad 0\leq T<t \leq 1.
\end{equation}
Thus $u(t)$ is bounded above and below, and so we conclude via the continuation principle for ODEs that the solution $u(t)$ must exist for all $t\in (0,1]$, that is, $T=0$. 

Next, integrating \eqref{HomeqB.1} in time, we see, with the help of the lower and upper bounds \eqref{Homprop1}-\eqref{Homprop2} and the triangle inequality,
that $u(t)$ satisfies the estimate
\begin{equation} \label{Homprop3}
|u(t_2)-u(t_1)| \lesssim \int^{t_2}_{t_1} \tau^{2\mu-1}\, d\tau  \lesssim t_2^{2\mu} - t_1^{2\mu}, \quad 0<t_1< t_2 \leq 1.
\end{equation}
From this, we conclude that the limit $\lim_{t\searrow 0} u(t)$ exists and $u(t)$ extends to a uniformly continuous function on $[0,1]$. Setting $t_2=t$ and sending $t_1\searrow 0$ in \eqref{Homprop3} gives
\begin{equation*}
|u(t)-u(0)| \lesssim t^{2\mu}, \quad 0<t\leq 1.
\end{equation*}
We further note that the inequality
\begin{equation*}
|u'\!(t)| \lesssim t^{2\mu-1},  \quad 0<t\leq 1,
\end{equation*}
follows directly from the bounds  \eqref{Homprop1}-\eqref{Homprop2} and the ODE \eqref{HomeqB.1}. To complete the proof, we observe,
by construction, that the trivial solution $W=0$ to \eqref{relEulE} 
determines via \eqref{cov1}, \eqref{zetatdef}, \eqref{cov2a}-\eqref{cov2c} and \eqref{Wdef} a homogeneous solution to the relativistic 
Euler equations \eqref{relEulA} given by \eqref{Homsol}.
\end{proof}

\subsection{Fuchsian form\label{Fform}}
To complete the transformation of the relativistic Euler equations into Fuchsian form, we let $u(t)$ denote one of the homogeneous solutions to the IVP   \eqref{HomeqB.1}-\eqref{HomeqB.2} from Proposition \ref{Homprop}. Then by construction, $\Fc_0=0$, and so,
\eqref{relEulE} reduces, see \eqref{Fcrep} and \eqref{Fc0def}, to 
\begin{equation} \label{relEulF}
\del{t}W + \Ac^I \del{I}W =-\frac{\mu}{t}\Pi W + t^{2\mu-1}\Gc.
\end{equation}
Applying the projection operator $\Pi$ to this equation and noting the $\Pi \Gc = 0$ by  \eqref{Gcdef} and \eqref{Pidef}, we get that 
\begin{equation*}
\del{t}(\Pi W) + \Pi\Ac^I \del{I}W =-\frac{\mu}{t}\Pi W,
\end{equation*}
which we observe can equivalently written as
\begin{equation}\label{relEulG}
\del{t}(t^{\mu}\Pi W) + t^{\mu}\Pi\Ac^I \del{I}W = 0.
\end{equation}
Next, applying $\Pi^\perp$, see \eqref{Piperpdef}, to \eqref{relEulF} shows, with the help of \eqref{Pirel}, that
\begin{equation*} 
\del{t}(\Pi^\perp W) + \Pi^\perp\Ac^I \del{I}W = t^{2\mu-1}\Pi^\perp\Gc.
\end{equation*}
Adding this equation to \eqref{relEulG} gives
\begin{equation} \label{relEulH}
\del{t}\Wb+t^\mu\Pi\Ac^I \del{I}W+\Pi^\perp\Ac^I \del{I}W =  t^{2\mu-1}\Pi^\perp\Gc
\end{equation}
where we have set
\begin{equation}\label{Wbdef}
\Wb := \Pi^\perp W+t^{\mu}\Pi W = (\zetat,w_1,t^{\mu}w_2,t^{\mu}w_3)^{\tr}.
\end{equation}

We then differentiate \eqref{relEulF} spatially to get
\begin{equation*}
\del{t}\del{J}W + \Ac^I \del{I}\del{J}W + \del{J}\Ac^I \del{I}W  = -\frac{\mu}{t}\Pi \del{J}W + t^{2\mu-1}\del{J}\Gc.
\end{equation*}
Setting
\begin{equation} \label{WbJdef}
\Wb\!_J := t^\mu \del{J}W = (t^\mu\del{J}\zetat,t^\mu \del{J}w_1,t^{\mu}\del{J}w_2,t^{\mu} \del{J} w_3)^{\tr},
\end{equation}
we can write this as
\begin{equation*} 
\del{t}\Wb\!_J + \Ac^I \del{I}\Wb\!_J + \del{J}\Ac^I \Wb_I  = \frac{\mu}{t}\Pi^\perp \Wb\!_J + t^{3\mu-1}\del{J}\Gc.
\end{equation*}
Multiplying on the left by $A^0$ and recalling the definitions \eqref{Acdef}, we find that $\Wb\!_J$ satisfies
\begin{equation}  \label{relEulI}
A^0\del{t}\Wb\!_J + A^I \del{I}\Wb\!_J = \frac{\mu}{t}A^0\Pi^\perp \Wb\!_J +  t^{3\mu-1}A^0\del{J}\Gc- A^0\del{J}\Ac^I \Wb_I.
\end{equation}
Additionally, using the definition \eqref{WbJdef}, we observe that \eqref{relEulH} can be written as
\begin{equation} \label{relEulJ}
\del{t}\Wb=-\Pi\Ac^I \Wb_I- t^{-\mu}\Pi^\perp\Ac^I \Pi \Wb_I + t^{2\mu-1}\Pi^\perp\Gc- t^{-\mu}\Pi^\perp\Ac^I \Pi^\perp \Wb_I.
\end{equation}

Finally, combining \eqref{relEulI} and \eqref{relEulJ} yields the system
\begin{equation} \label{relEulK}
\Asc^0\del{t}\Wsc + \Asc^I \del{I}\Wsc = \frac{\mu}{t}\Asc^0\Pbb \Wsc + \Fsc_0 +\Fsc_1
\end{equation}
where
\begin{align}
\Wsc &= \begin{pmatrix} \Wb \\ \Wb\!_J \end{pmatrix}, \label{Wscdef} \\
\Asc^0 &= \begin{pmatrix} \id & 0 \\ 0 & A^0 \end{pmatrix}, \label{Asc0def} \\
\Asc^I &=  \begin{pmatrix} 0 & 0 \\ 0 & A^I \end{pmatrix}, \label{AscIdef} \\
\Pbb &=  \begin{pmatrix} 0 & 0 \\ 0 & \Pi^\perp \end{pmatrix}, \label{Pbbdef} \\
\Fsc_0 &=\begin{pmatrix} -\Pi\Ac^I \Wb_I- t^{-\mu}\Pi^\perp\Ac^I \Pi \Wb_I + t^{2\mu-1}\Pi^\perp\Gc \\ 
t^{3\mu-1}A^0\del{J}\Gc- A^0\del{J}\Ac^I \Wb_I\end{pmatrix} \label{Fsc0def}
\intertext{and}
\Fsc_1 & =\begin{pmatrix} - t^{-\mu}\Pi^\perp\Ac^I \Pi^\perp \Wb_I \\ 0\end{pmatrix}. \label{Fsc1def}
\end{align}
The point of this system, as will be established in the proof of Theorem \ref{mainthm}, is that it is now of a suitable Fuchsian form to which we can apply the existence theory from \cite{BOOS:2020}. This will allow us to establish the future stability of nonlinear perturbations of the homogeneous solutions
to relativistic Euler equations that are defined by \eqref{Homsol}.

\subsection{Coefficient properties\label{coeff}}
We now turn to showing that the coefficients of the system \eqref{relEulK} satisfy the required properties needed to apply 
 the existence theory from \cite{BOOS:2020} in the proof of Theorem \ref{mainthm}.
To begin, we define
\begin{equation} \label{bvarsdef}
\tb = t^{2\mu}, \AND \wb_\Lambda = t^\mu w_\Lambda, \quad \Lambda=2,3,
\end{equation}
and observe from \eqref{wbr1def}-\eqref{phidef}, \eqref{psidef} and \eqref{A0rep} that the matrix $A^0$ can be treated as a map depending on the variables \eqref{bvarsdef}, that is,
\begin{equation} \label{A0smooth}
A^0 = A^0(\tb,\wbr_1,\wb_2,\wb_3),
\end{equation}
where for each $R>0$ there exists constants $r,\omega >0$ such that $A^0$  is smooth
on the domain defined by
\begin{equation} \label{smoothdom}
(\tb,\wbr_1,\wb_2,\wb_3) \in (-r,2) \times (-R,R) \times (-R,R)\times (-R,R),
\end{equation}
and satisfies
\begin{equation} \label{A0lb}
A^0(\tb,\wbr_1,0,0) \geq \omega \id 
\end{equation}
for all $(\tb,\wbr_1)\in (-\rho,2)\times (-R,R)$. Differentiating $A^0$ with respect to $t$ then shows, with the help of  \eqref{wbr1def}, \eqref{Wbdef}
and \eqref{bvarsdef}, that
\begin{align} 
\del{t}A^0 &= DA^0(\tb,\wbr_1,\wb_2,\wb_3) \begin{pmatrix} 2\mu t^{2\mu-1} \\ u'(t)+\del{t}w_1\\ \del{t}\wb_2\\ \del{t}\wb_3 \end{pmatrix}
\notag \\
&= DA^0(\tb,\wbr_1,\wb_2,\wb_3) \left(\begin{pmatrix} 2\mu t^{2\mu-1} \\ u'(t)\\ 0 \\ 0 \end{pmatrix}
+ \Pc_1\del{t}\Wb \right) \label{dtA0}
\end{align}
where 
\begin{equation*} 
\Pc_1 = \begin{pmatrix} 0 & 0 & 0 & 0 \\ 0 & 1 & 0 & 0 \\ 0 & 0 & 1 & 0 \\ 0 & 0 & 0 & 1 \end{pmatrix},
\end{equation*}
$\del{t}\Wb$ can be computed using the equation of motion \eqref{relEulJ}, and $u'(t)$ is bounded by \eqref{Hombounds}.

Next, setting
\begin{equation} \label{wh1def}
\wh_1 = t^\mu e^{-2 \wbr_1}, 
\end{equation}
it follows from \eqref{wbr1def}-\eqref{etadef}, \eqref{Ac1rep}-\eqref{Ac3rep}, \eqref{psidef}-\eqref{chidef} and \eqref{bvarsdef} that 
we can express the matrices $\Ac^I$ as
\begin{equation} \label{AcIsmooth}
\Ac^I = \Ac^I_1(\wh_1,\wb_2,\wb_3)+ t^\mu \Ac^I_2(\tb,\wbr_1,\wb_2,\wb_3)+ t^{2\mu} \Ac^I_3(\tb,\wbr_1,\wb_2,\wb_3)
\end{equation}
where the $\Ac^I_2$, $\At^I_3$ are smooth on the domain \eqref{smoothdom} and the $\Ac^I_1$ are smooth on the domain
defined by
\begin{equation*}
(\wh_1,\wb_2,\wb_3) \in (-R,R)\times (-R,R)\times (-R,R).
\end{equation*}
 It is also not difficult to verify from \eqref{Ac1rep}-\eqref{Ac3rep} that the $\Ac^I_1$
 satisfy
 \begin{equation} \label{PiperpAcIPi}
 \Pi^\perp \Ac^I_1 \Pi = 0.
 \end{equation}
Differentiating the matrices $\Ac^I$ spatially, we get from \eqref{wbr1def}, \eqref{WbJdef}, \eqref{bvarsdef}, \eqref{wh1def} and
\eqref{AcIsmooth} that
\begin{align}
\del{J}\Ac^I &= D\Ac^I_1(\wh_1,\wb_2,\wb_3)\begin{pmatrix}-2 e^{-2 \wbr_1} t^\mu \del{J}w_1 \\ t^\mu \del{J}w_2 \\
 t^\mu \del{J}w_2
\end{pmatrix}\notag \\
&+ t^\mu D\Ac^I_2(\tb,\wbr_1,\wb_2,\wb_3)\begin{pmatrix} 0 \\ \del{J}w_1\\ t^\mu \del{J}w_2
\\ t^\mu\del{J} w_3 \end{pmatrix} +t^{2\mu} D\Ac^I_3(\tb,\wbr_1,\wb_2,\wb_3)\begin{pmatrix} 0 \\ \del{J}w_1\\ t^\mu \del{J}w_2
\\ t^\mu\del{J} w_3 \end{pmatrix} \notag \\
& =  \Bigl(D\Ac^I_1(\wh_1,\wb_2,\wb_3)\Pc_2 + D\Ac^I_2(\tb,\wbr_1,\wb_2,\wb_3)\Pc_3+t^\mu  D\Ac^I_2(\tb,\wbr_1,\wb_2,\wb_3)\Pc_3\Bigr) \Wb\!_J, \label{dJAcI}
\end{align}
where
\begin{equation*}
\Pc_2 = \begin{pmatrix} 0 & -2 e^{-2 \wbr_1} & 0 & 0 \\ 0 & 0 & 1 & 0 \\ 0 & 0 & 0 & 1 \end{pmatrix} \AND
\Pc_3 = \begin{pmatrix} 0 & 0 & 0 & 0 \\ 0 & 1 & 0 & 0 \\ 0 & 0 & t^\mu & 0 \\ 0 & 0 & 0 & t^\mu \end{pmatrix}.
\end{equation*}

We further observe from \eqref{A0rep} and \eqref{Pidef}-\eqref{Piperpdef} that $A^0$ satisfies
\begin{equation*}
[\Pi^\perp,A^0] = 0
\end{equation*}
and
\begin{equation*}
\Pi^\perp A^0 \Pi = \Pi A^0 \Pi^\perp = 0.
\end{equation*}
From the definitions \eqref{Asc0def} and \eqref{Pbbdef}, it is then clear that $\Asc^0$ satisfies
 \begin{equation} \label{Asc0Pbbcom}
[\Pbb,\Asc^0] = 0
\end{equation}
and
\begin{equation} \label{PbbAsc0Pbbperp}
\Pbb^\perp \Asc^0 \Pbb = \Pbb \Asc^0 \Pbb^\perp = 0,
\end{equation}
where 
\begin{equation} \label{Pbbperpdef}
\Pbb^\perp = \id -\Pbb.
\end{equation}
Additionally, it follows immediately from \eqref{Pidef}-\eqref{Pirel} that $\Pbb$ satisfies
\begin{equation} \label{Pbbrel}
\Pbb^2 = \Pbb,  \quad \Pbb^{\tr} = \Pbb,  \quad \del{t}\Pbb = 0 \AND \del{I} \Pbb = 0,
\end{equation}
while the symmetry of the matrices $\Asc^i$,  that is,
\begin{equation} \label{Ascisym}
(\Asc^i)^{\tr} = \Asc^i
\end{equation}
is an immediate consequence of the definitions \eqref{Bt0def}-\eqref{BtIdef}, \eqref{Aidef}, and \eqref{Asc0def}-\eqref{AscIdef}.

\begin{rem} \label{dtA0rem}
From the definitions  \eqref{wbr1def}, \eqref{Gcdef}, \eqref{Wbdef}, \eqref{WbJdef}, \eqref{Wscdef}, \eqref{Pbbdef}, \eqref{bvarsdef}
and \eqref{wh1def}, the evolution equation  \eqref{relEulJ}, the estimates \eqref{Hombounds} for $u(t)$ and $u'(t)$, the derivative formula
\eqref{dtA0}, and the smoothness properties \eqref{A0smooth} and \eqref{AcIsmooth} of the matrices $A^0$ and $\Ac^I$, respectively, and the identity
\eqref{PiperpAcIPi}, it is not difficult to verify, for $\mu$ satisfying
$0\leq \mu \leq 1/2$ and $R>0$ small enough, that there exists
constants $\theta, \beta>0$ such that $\del{t}A^0$ is bounded by
\begin{equation*}
|\del{t}A^0| \leq \theta t^{2\mu-1} + t^{\mu-1}\beta |\Pbb \Wsc|
\end{equation*}
for all $(t,\Wsc)\in [0,1]\times B_R(\Rbb^{16})$. Furthermore, from the formulas \eqref{Fsc0def}-\eqref{Fsc1def}, it is also clear that
\begin{equation*}
\Pbb\Fsc_1=0
\end{equation*}
and there exists a constant $\lambda>0$ such that $\Pbb^\perp\Fsc_1$ is bounded by
\begin{equation*}
|\Pbb^\perp\Fsc_1| \leq t^{\mu-1}\lambda|\Pbb \Wsc|
\end{equation*}
for all $(t,\Wsc)\in [0,1]\times B_R(\Rbb^{16})$, while $\Fsc_0$ is bounded by
\begin{equation*}
|\Fsc_0| \lesssim t^{2\mu-1}|\Wsc| 
\end{equation*}
for all $(t,\Wsc)\in [0,1]\times B_R(\Rbb^{16})$.

By similar considerations, we see for $\mu$ satisfying $1/2<\mu <1$, that  $\del{t}A^0$,  $\Pbb^\perp \Fsc_1$ and $\Fsc_0$ are bounded by
\begin{align*}
|\del{t}A^0| &\leq \theta t^{1-2\mu} + t^{-\mu}\beta |\Pbb \Wsc|, \\
|\Pbb^\perp\Fsc_1| &\leq t^{-\mu}\lambda |\Pbb \Wsc|
\intertext{and}
|\Fsc_0| &\lesssim t^{1-2\mu}|\Wsc|,
\end{align*}
respectively,
for all $(t,\Wsc)\in [0,1]\times B_R(\Rbb^{16})$.
\end{rem}

\section{Future stability\label{stability}}
We are now ready to establish the future stability of nonlinear perturbations of the homogeneous solutions \eqref{Homsol} of relativistic Euler equations.
\begin{thm} \label{mainthm}
Suppose $k\in\Zbb_{>3/2+3}$, $1/3<K < 1/2$, $\mu = (3K-1)/(1-K)$,  $\sigma > 0$, $u_0\in \Rbb$, $u \in C^\infty((0,1])\cap C^0([0,1])$ is the unique solution to the IVP \eqref{HomeqB.1}-\eqref{HomeqB.2}, and
$\zetat_0, w^0_J \in H^{k+1}(\Tbb^3)$.
Then for $\delta>0$ small enough, there exists a unique solution
\begin{equation*}
W=(\zetat,w_J)^{\tr} \in C^0\bigl((0,1], H^{k+1}(\Tbb^3,\Rbb^4)\bigr)\cap C^1\bigl((0,1],H^{k}(\Tbb^3,\Rbb^4)\bigr)
\end{equation*}
to (see \eqref{relEulD} and \eqref{Fcdef}) the initial value problem 
\begin{align}
A^0 \del{t}W + A^I \del{I}W &= A^0 \Fc &&  \text{in $(0,1]\times \Tbb^3$,} \label{relEulL1} \\
W  &= (\zetat_0, w^0_J)^{\tr} && \text{in $\{1\}\times \Tbb^3$,} \label{relEulL2}
\end{align}
provided that 
\begin{equation*}
\biggl(\norm{\zetat_0}_{H^{k+1}}^2+\sum_{J=1}^3\norm{w^0_J}_{H^{k+1}}^2\biggr)^{\frac{1}{2}}\leq \delta.
\end{equation*}
Moreover, 
\begin{enumerate}[(i)]
\item $W=(\zetat,w_J)^{\tr}$ satisfies the energy estimate
\begin{equation*}
\Ec(t) + \int_t^1 \tau^{2\mu-1}\bigl(\norm{D\zetat(\tau)}_{H^k}^2+\norm{Dw_1(\tau)}_{H^k}^2\bigr)\,d\tau \lesssim \norm{\zetat_0}_{H^{k+1}}^2+\sum_{J=1}^3\norm{w^0_J}_{H^{k+1}}^2
\end{equation*}
for all $t\in (0,1]$
where\footnote{Here, the norm $\norm{Df}_{H^k}$ is defined by  $\norm{Df}^2_{H^k}= \sum_{J=1}^3 \norm{\del{J}f}^2_{H^k}$.}
\begin{equation*}
\Ec(t)=\norm{\zetat(t)}_{H^k}^2+\norm{w_1(t)}_{H^k}^2+t^{2\mu}\Bigl(\norm{D\zetat(t)}_{H^k}^2+\norm{Dw_1(t)}_{H^k}^2+\norm{w_2(t)}_{H^{k+1}}^2+\norm{w_3(t)}_{H^{k+1}}^2\Bigr),
\end{equation*}
\item there exists functions $\zetat_*, w_1^* \in H^{k-1}(\Tbb^3)$ and $\wb_2^*,\wb_3^* \in H^{k}(\Tbb^3)$ such that the estimate
\begin{align*}
\bar{\Ec}(t) \lesssim
 t^{\mu-\sigma}
\end{align*}
holds for all $t\in (0,1]$ where
\begin{equation*}
\bar{\Ec}(t)=\norm{\zetat(t) - \zetat_*}_{H^{k-1}}+\norm{w_1(t) - w_1^*}_{H^{k-1}}
+\norm{t^\mu w_2(t) - \wb_2^*}_{H^{k}}+\norm{t^\mu w_3(t) - \wb_3^*}_{H^{k}},
\end{equation*}
\item and $u$ and $W=(\zetat,w_J)^{\tr}$ determine a solution of the relativistic Euler equations \eqref{relEulA} on the
spacetime region $M=(0,1]\times \Tbb^3$ via the formulas 
\begin{align}
\rho &= \frac{\rho_c t^{\frac{2(1+K)}{1-K}} e^{(1+K)\zetat}}{(t^{2\mu}+ e^{2(u+w_1)})^{\frac{1+K}{2}}}, \label{relEulsol.1}\\
\vt^0 &= -t^{1-\mu}\sqrt{e^{2 (u+w_1)}+t^{2 \mu} },\label{relEulsol.2}\\
\vt^1 &=t^{1-\mu }\biggl(  \frac{e^{u+w_1}}{\sqrt{ (t^{\mu}w_2-t^{\mu}w_3)^2+(t^{\mu}w_2+t^{\mu}w_3)^2+1}} \biggr), \label{relEulsol.3} \\
\vt^2 &= t^{1-\mu }\biggl( \frac{(t^{\mu}w_2+t^{\mu}w_3) e^{u+w_1}}{\sqrt{ (t^{\mu}w_2-t^{\mu}w_3)^2+(t^{\mu}w_2+t^{\mu}w_3)^2+1}}\biggr)\label{relEulsol.4}
 \intertext{and}
\vt^3 &= t^{1-\mu }\biggl( \frac{(t^{\mu}w_2-t^{\mu}w_3) e^{u+w_1}}{\sqrt{ (t^{\mu}w_2-t^{\mu}w_3)^2+(t^{\mu}w_2+t^{\mu}w_3)^2+1}}\biggr).  \label{relEulsol.5}
\end{align}
\end{enumerate}
\end{thm}
\begin{proof}
By \eqref{Bt0def}-\eqref{BtIdef} and \eqref{Aidef}, we know that the matrices $A^i$ are symmetric. Furthermore, from the analysis carried out in Section \ref{coeff}, we know that the maps $A^i$ and $\Fc$ depend smoothly on the variables $(t,\zeta,w_J)$ for $t\in (0,1]$ and $(\zeta,w_J)$ in an open neighborhood of zero, and that the matrix $A^0$ is
positive definite. This shows that the system \eqref{relEulL1} is symmetric hyperbolic. Since $k\in\Zbb_{>3/2+3}$ and 
\begin{equation*}
W_0 :=(\zetat_0, w^0_J)^{\tr}\in H^{k+1}(\Tbb^3,\Rbb^4))
\end{equation*}
by assumption, we can appeal to standard local-in-time existence and uniqueness theorems and the continuation principle for symmetric hyperbolic systems, see Propositions 1.4, 1.5 and 2.1 from \cite[Ch.~16]{TaylorIII:1996},  to conclude that there exists a maximal time $T_* \in [0,1)$ such that
the IVP \eqref{relEulL1}-\eqref{relEulL2} admits a unique solution
\begin{equation*}
W=(\zetat,w_J)^{\tr} \in C^0((T_*,1], H^{k+1}(\Tbb^3,\Rbb^4))\cap C^1((T_*,1],H^{k}(\Tbb^3,\Rbb^4).
\end{equation*}
We also know from the computations carried out in Section \ref{Ftrans} that
\begin{equation} \label{Wscvar}
\Wsc = (\Wb,\Wb\!_J),
\end{equation}
where $\Wb$ and $\Wb\!_J$ are determined from the solution $W$ via the formulas \eqref{Wbdef} and \eqref{WbJdef}, respectively,
will solve the IVP
\begin{align} 
\Asc^0\del{t}\Wsc + \Asc^I \del{I}\Wsc &= \frac{\mu}{t}\Asc^0\Pbb \Wsc + \Fsc_0 + \Fsc_1 && 
\text{in $(T_*,1]\times \Tbb^3$,} 
\label{relEulM1} \\
\Wsc &= \Wsc_0 := (W_0,\del{J}W_0)^{\tr} &&\text{in $\{1\}\times \Tbb^3$.} \label{relEulM2}
\end{align}
We further observe that if the initial data $W_0$ is chosen to satisfy $\norm{W_0}^{H^{k+1}} \leq \delta$, then
\begin{equation*}
\norm{\Wsc_0}_{H^k} \lesssim \norm{W_0}_{H^{k+1}} \leq \delta.
\end{equation*}

On the other hand, we can view \eqref{relEulM1} as an equation for the variables $\Wsc=(\Wb,\Wb_J)$, with $\Wb=(w_1,\wb_2,\wb_3)$ and $\Wb=(w_{1,J},\wb_{1,J},\wb_{1,J})$, where the maps $A^0$, $A^I=A^0\Ac^I$ and $\Fsc_0$, $\Fsc_1$ depend on the variables $(t,\Wb)$ and $\Wsc$ respectively; see Section \ref{coeff} above. 
Then from \textbf{(i)} the smoothness properties \eqref{A0smooth}, \eqref{AcIsmooth} and the identity \eqref{PiperpAcIPi} satisfied by the matrices $A^0$ and $\Ac^I$,
\textbf{(ii)} the
derivative formulas \eqref{dtA0} and \eqref{dJAcI},  \textbf{(iii)} the variable definitions \eqref{wbr1def}, \eqref{Wbdef}, \eqref{WbJdef}, \eqref{bvarsdef},
and \eqref{wh1def}, \textbf{(iv)} the properties of the homogeneous solution $u(t)$ as given by Proposition \ref{Homprop}, \textbf{(v)} the
assumption $1/3<K < 1/2$, which, in particular, implies that $0<\mu <1$, and \textbf{(vi)} the properties \eqref{Asc0Pbbcom}-\eqref{Ascisym} of the matrices
$\Asc^i$ and $\Pbb$, it is not difficult to verify using the definitions \eqref{Wscdef}-\eqref{Fsc1def} that, for $R>0$ chosen small enough, 
there exists, see also Remark \ref{dtA0rem}, constants $\theta,\gamma_1=\gammat_1,\gamma_2=\gamma_2,\lambda_2>0$ and
$\beta_0,\beta_2,\beta_4,\beta_6>0$, where the $\beta_q$  
can be chosen as small as we like by shrinking $R>0$ if necessary, such that system \eqref{relEulM1}
satisfies all the assumptions from Section 3.4 of \cite{BOOS:2020}  for following choice of constants:  
\begin{gather*}
p=\begin{cases}2\mu & \text{if $0<\mu\leq 1/2$} \\
2(1-\mu) & \text{if $1/2<\mu < 1$} \end{cases}, \\
 \kappa=\kappat=\mu
 \intertext{and}
\beta_1=\beta_3=\beta_5=\beta_7=\lambda_1=\lambda_3= \alpha=0.
\end{gather*} As discussed in \cite[\S 3.4]{BOOS:2020}, this implies that under the time transformation\footnote{By our conventions, the time variable $t$ is assumed to be positive as opposed to \cite{BOOS:2020} where it is taken to
be negative. This
causes no difficulties since one can change between these two conventions by using the simple time transformation $t\rightarrow -t$.} $t \mapsto t^p$, the transformed version of \eqref{relEulK} will satisfy all of the assumptions from Section 3.1 of \cite{BOOS:2020}.  Moreover, since the $\Asc^I$ have a regular limit as $t\searrow 0$ (equivalently as $\tau\searrow 0$), the constants $\btt$ and $\tilde{\btt}$ from Theorem 3.8 of \cite{BOOS:2020} will satisfy $\btt=\tilde{\btt}=0$, and
consequently, the constant\footnote{In the article \cite{BOOS:2020}, this constant is denoted by $\zeta$, but since we are already using $\zeta$ to denote the
modified fluid density, we will refer to this parameter as $\mathfrak{z}$.}  $\mathfrak{z}$ that is involved in determining the decay is given by
\begin{equation*}
\mathfrak{z}= \kappa - \Half \gammat_1(\beta_1+(k-1)\tilde{\btt}) = \mu.
\end{equation*}
We can therefore apply Theorem 3.8 from \cite{BOOS:2020}
to the time transformed version of \eqref{relEulM1} as described in \cite[Section 3.4]{BOOS:2020} to deduce, for $\delta>0$ chosen small enough and
the initial data satisfying $\norm{\Wsc(0)}_{H^k}< \delta$, the existence of a unique
solution
\begin{equation*}
\Wsc^* \in C^0\bigl((0,1],H^k(\Tbb^3,\Rbb^{16})\bigr)\cap L^\infty\bigl((0,1],H^k(\Tbb^3,\Rbb^{16}))\bigr)\cap 
C^1\bigl((0,1],H^{k-1}(\Tbb^3,\Rbb^{16})\bigr)
\end{equation*}
to the IVP \eqref{relEulM1}-\eqref{relEulM2} with the following properties:
\begin{enumerate}[(1)]
\item The limit $\lim_{t\searrow 0} \Pbb^\perp \Wsc^*$, denoted $\Pbb^\perp \Wsc^*(0)$, exists in $H^{k-1}(\Tbb^3,\Rbb^{16})$.
\item The solution satisfies the energy estimate
\begin{equation}\label{eestA}
\norm{\Wsc^*(t)}_{H^k}^2 + \int_{t}^1 \frac{1}{\tau} \norm{\Pbb \Wsc^*(\tau)}_{H^k}^2\, d\tau   \lesssim \norm{\Wsc_0}_{H^k}^2
\end{equation}
for all $t\in (0,1]$.
\item The solution decays as $t\searrow 0$ according to
\begin{gather}
\label{decayA1}
\norm{\Pbb \Wsc^*(t)}_{H^{k-1}} \lesssim 
t^{\mu-\sigma} 
\intertext{and}
\label{decayA2}
\norm{\Pbb^\perp \Wsc^*(t) - \Pbb^\perp \Wsc^*(0)}_{H^{k-1}} \lesssim
 t^{\mu-\sigma} 
\end{gather}
for all $t\in (0,1]$.
\end{enumerate}

By uniqueness, the two solutions $\Wsc$ and $\Wsc^*$ to the IVP \eqref{relEulM1}-\eqref{relEulM2} must coincide on their common
domain of definition, and so, we have
\begin{equation*}
\Wsc(t)=\Wsc^*(t), \quad T_*<t \leq 1.
\end{equation*}
But this implies via  \eqref{Wscvar}, the energy estimate \eqref{eestA}, and Sobolev's inequality \cite[Ch.~13, Prop 2.4]{TaylorIII:1996} that
\begin{equation*}
\norm{\Wb(t)}_{W^{1,\infty}} \lesssim \norm{\Wb(t)}_{H^k} \leq \norm{\Wsc(t)}_{H^{k-1}} \lesssim \norm{\Wsc_0},
\quad T^*<t\leq 1.
\end{equation*}
By choosing the initial data $\Wsc_0$ so that $\norm{\Wsc_0}_{H^k}$ is sufficiently small, we can then ensure that
\begin{equation*}
\norm{\Wb(t)}_{W^{1,\infty}} \leq \frac{R}{2}
\quad T^*<t\leq 1,
\end{equation*}
where $R>0$ is as defined in Section \ref{coeff}, which, in particular,
is enough to guarantee that the coefficients $A^i$ and $\Fc$ of \eqref{relEulL1} remain well-defined. By the continuation principle and the maximality of $T_*$, we conclude that $T_*=0$, and hence that
\begin{equation*}
\Wsc(t)=\Wsc^*(t), \quad 0< t\leq  1.
\end{equation*}
From this, the definitions \eqref{Pidef}, \eqref{Piperpdef}, \eqref{Wbdef}, \eqref{WbJdef},  \eqref{Pbbdef} and \eqref{Wscvar},  and  the energy estimate \eqref{eestA}, it is then straightforward to verify
\begin{equation*}
\Ec(t) + \int_t^1 \tau^{2\mu-1}\bigl(\norm{D\zetat(\tau)}_{H^k}^2+\norm{Dw_1(\tau)}_{H^k}^2\bigr)\,d\tau \lesssim \norm{W_0}_{H^k}^2,
\quad 0<t\leq 1,
\end{equation*}
where
\begin{equation*}
\Ec(t)=\norm{\zetat(t)}_{H^k}^2+\norm{w_1(t)}_{H^k}^2+t^{2\mu}\Bigl(\norm{D\zetat(t)}_{H^k}^2+\norm{Dw_1(t)}_{H^k}^2+\norm{w_2(t)}_{H^{k+1}}^2+\norm{w_3(t)}_{H^{k+1}}^2\Bigr).
\end{equation*}
Furthermore, from the decay estimate \eqref{decayA2} and the definition \eqref{Pbbperpdef}, we obtain the existence of functions $\zetat_*, w_1^* \in H^{k-1}(\Tbb^3)$ and $\wb_2^*,\wb_3^* \in H^{k}(\Tbb^3)$ such that the estimate
\begin{align*}
\bar{\Ec}(t) \lesssim
 t^{\mu-\sigma}
\end{align*}
holds for all $t\in (0,1]$, where
\begin{equation*}
\bar{\Ec}(t)=\norm{\zetat(t) - \zetat_*}_{H^{k-1}}+\norm{w_1(t) - w_1^*}_{H^{k-1}}
+\norm{t^\mu w_2(t) - \wb_2^*}_{H^{k}}+\norm{t^\mu w_3(t) - \wb_3^*}_{H^{k}}.
\end{equation*}
To complete the proof, we recall from  \eqref{conformal}, \eqref{cov1}, \eqref{v0def} and \eqref{cov2a}-\eqref{cov2c}, that $u$ and
$W=(\zetat,w_J)^{\tr}$ determine a solution of the relativistic Euler equations \eqref{relEulA} on the spacetime region $M=(0,1]\times \Tbb^3$
via the formulas  \eqref{relEulsol.1}-\eqref{relEulsol.5}.
\end{proof}

\section{$\Tbb^2$-symmetric future stability\label{symsec}} 
In this section, we focus on solutions of the relativistic Euler equations that are independent of the coordinates $(x^2,x^3)\in \Tbb^2$, or in other words,
admit a $\Tbb^2$-symmetry. To find such solutions, we set
\begin{equation} \label{symvars}
\zetat=\ztt(t,x^1), \quad w_1 = \wtt(t,x^1) \AND w_2=w_3=0,
\end{equation}
and observe, with the help of the \eqref{Wdef}, \eqref{wbr1def}-\eqref{xidef}, \eqref{Acdef}, \eqref{Fcdef} and \eqref{A0rep}-\eqref{Pidef},
that this ansatz leads to a consistent reduction of \eqref{relEulD} to a symmetric hyperbolic equations for the variables $(\ztt,\wtt)$ in $1+1$ dimensions
given by
\begin{equation} \label{symEulA}
\Att^0 \del{t}\Wtt+ \Att^1 \del{1}\Wtt = \Ftt
\end{equation}
where
\begin{gather}
\Wtt = (\ztt,\wtt)^{\tr}, \label{Wttdef} \\
\Att^0 =\begin{pmatrix}
 K & 0  \\
 0 & \frac{t^{2 \mu } e^{2 (u+\wtt)}+(1-K) e^{4
   (u+\wtt)}}{(t^{2 \mu }+e^{2
   (u+\wtt)})^2}  
\end{pmatrix} , \label{Att0def}\\
\Att^1 =\frac{1}{\sqrt{e^{2 (u+\wtt)}+t^{2
   \mu } }}\begin{pmatrix}
 -K e^{u+\wtt} & -\frac{K
   t^{2\mu } e^{u+\wtt}}{t^{2 \mu }+e^{2
   (u+\wtt)}} \\
 -\frac{K t^{2\mu } e^{u+\wtt}}{t^{2 \mu
   }+e^{2 (u+\wtt)}} & \frac{
   (2 K-1) t^{2 \mu } e^{3
   (u+\wtt)}+(K-1) e^{5
   (u+\wtt)}}{ (t^{2 \mu }+e^{2
   (u+\wtt)})^2} 
\end{pmatrix}  ,\label{Att1def}\\
\Ftt =\begin{pmatrix}
 0 \\
 -\frac{ t^{2\mu -1 }K (3 K-1) (e^{2 \wtt}-1) e^{4
   u+2 \wtt}}{(t^{2 \mu }-(K-1) e^{2
   u}) (t^{2 \mu }+e^{2
   (u+\wtt)})^2} 
   \end{pmatrix}
   ,\label{Fttdef} 
\end{gather}
and in deriving this equation, we have assumed, as above, that $u=u(t)$ solves the IVP \eqref{HomeqB.1}-\eqref{HomeqB.2}.
 
The system \eqref{symEulA} is almost regular in that $\Att^0(t,\wtt)$, $\Att^1(t,\wtt)$ and $\Ftt(t,\wtt)$ are smooth in $(t,\wtt)$ for $(t,\wtt)\in (0,1]\times \Rbb$ and 
$\Att^0$ and $\Att^1$ are, for any $R>0$, uniformly bounded for
$(t,\wtt)\in (0,1]\times [-R,R]$  by virtue of the assumption $1/3<K<1$, which implies that $\mu>0$. The slight difficulty in establishing existence is that  
$\del{t}A^0(t,\wtt)$  and $\Ftt(t,\wtt)$ are not bounded as $t\searrow 0$ for all $K\in (1/3,1)$.
However, the worst that these coefficients can diverge is like $t^{2\mu -1}$, which is always integrable since $\mu>0$. As we shall
see in the proof of the following theorem, this integrability allows us to modify standard local-in-time 
existence results in a straightforward fashion to establish the existence of solutions of \eqref{symEulA} on $(0,1]\times \Tbb^1$ under a suitable small initial data assumption.

\begin{thm} \label{symthm}
Suppose $k\in \Zbb_{>1/2+1}$, $1/3<K<1$, $\mu=(3K-1)/(1-K)$,  $u_0\in \Rbb$ and $u\in C^\infty((0,1])\cap C^0([0,1])$ is the unique solution
to the IVP \eqref{HomeqB.1}-\eqref{HomeqB.2}, and $\ztt_0,\wtt_0 \in H^k(\Tbb^1)$. Then for $\delta>0$ small enough, there exists a unique solution
\begin{equation*}
\Wtt=(\ztt,\wtt)^{\tr} \in C^0\bigl((0,1],H^k(\Tbb^1,\Rbb^2)\bigr)\cap C^1\bigl((0,1],H^{k-1}(\Tbb^1,\Rbb^2)\bigr)
\end{equation*} 
of the IVP
\begin{align}
\Att^0 \del{t}\Wtt+ \Att^1 \del{1}\Wtt &= \Ftt && \text{in $(0,1]\times \Tbb^1$,}\label{globalS1}\\
\Wtt &= (\ztt_0,\wtt_0)^{\tr} && \text{in $\{1\}\times \Tbb^1,$} \label{globalS2} 
\end{align}
provided that 
\begin{equation*}
\bigl(\norm{\ztt_0}^2_{H^k}+\norm{\wtt_0}^2_{H^k} \bigr)^{\frac{1}{2}} \leq \delta.
\end{equation*}
Moreover,
\begin{enumerate}[(i)]
\item the solution and its time derivative are bounded by
\begin{equation*}
\norm{\Wtt(t)}_{H^k} \lesssim 1 \AND \norm{\del{t}\Wtt(t)}_{H^{k-1}} \lesssim 1+t^{2\mu-1} 
\end{equation*}
respectively, for all $t\in (0,1]$,
\item there exist functions $\ztt_*,\wtt_*\in H^{k-1}(\Tbb^1)$ such that
\begin{equation*}
\norm{\ztt(t)-\ztt_*}_{H^{k-1}}+\norm{\wtt(t)-\wtt_*}_{H^{k-1}} \lesssim t+ t^{2\mu}
\end{equation*}
for all $t\in (0,1]$,
\item and $u$ and $\Wtt=(\ztt,\wtt)^{\tr}$ determine a solution of the relativistic Euler equations \eqref{relEulA} on the spacetime region $M=(0,1]\times \Tbb^3$ via
the formulas
\begin{align*}
\rho &= \frac{\rho_c t^{\frac{2(1+K)}{1-K}} e^{(1+K)\ztt}}{(t^{2\mu}+ e^{2(u+\wtt)})^{\frac{1+K}{2}}}, \\ 
\vt^0 &= -t^{1-\mu}\sqrt{e^{2 (u+\wtt)}+t^{2 \mu} }, \\ 
\vt^1 &=t^{1-\mu }e^{u+\wtt} 
\intertext{and}
\vt^2 &=\vt^3 = 0. 
\end{align*}
\end{enumerate}
\end{thm}
\begin{proof}
Since $K\in (1/3,1)$, the inequality $\mu>0$ holds, and so, fixing $R>0$, we observe from \eqref{Att0def} that there exists a constant
$\gamma>0$
such that
\begin{equation}\label{globalP2}
\frac{1}{\gamma}\id \leq \Att^0(t,\wtt)  \leq \gamma\id 
\end{equation}
for all $(t,\wtt)\in (0,1]\times [-R,R]$. From this inequality and the smooth dependence, see \eqref{Att0def}-\eqref{Fttdef}, of $\Att^0$, $\Att^1$ and $\Ftt$ on $(t,\wtt)$ for $(t,\wtt)\in (0,1]\times \Rbb$,
it follows that the system \eqref{globalS1} is symmetric hyperbolic. Consequently, fixing
$k \in \Zbb_{> 1/2+1}$ and choosing initial data $\ztt_0,\wtt_0\in H^{k}(\Tbb^1)$ satisfying 
\begin{equation} \label{globalP3}
\norm{\Wtt(1)}_{H^k} = \sqrt{\norm{\ztt_0}_{H^k}^2 + \norm{\wtt_0}_{H^k}^2} \leq \delta 
\end{equation} 
for some $\delta>0$, we know from
standard local-in-time existence and uniqueness theorems and the continuation principle for symmetric hyperbolic systems, see Propositions 1.4, 1.5 and 2.1 from \cite[Ch.~16]{TaylorIII:1996}, that there exists 
a unique solution 
\begin{equation*} 
\Wtt=(\ztt, \wtt )^{\tr} \in C^0\bigl((T^*,1],H^{k}(\Tbb^1)\bigr)\cap C^1 \bigl((T^*,1],H^{k-1}(\Tbb^1)\bigr)
\end{equation*}
to \eqref{globalS1} satisfying the initial condition \eqref{globalS2}
for some time $T^*\in [0,1)$, which we can take to be maximal.

Next, applying  $\Att^0\del{1}^\ell (\Att^0)^{-1}$  to  \eqref{globalS1} gives
\begin{equation}  \label{globalP7}
\Att^0 \del{t}\del{1}^\ell \Wtt+ \Att^1 \del{1} \del{1}^\ell \Wtt= \Ftt_\ell, \quad \ell =0,1,\dots,k, 
\end{equation}
where
\begin{equation} \label{globalP9}
\Ftt_\ell =-\Att^0[\del{1}^\ell,(\Att^0)^{-1}\Att^1]\del{1}\Wtt + \Att^0\del{1}^\ell\bigl((\Att^0)^{-1}\Ftt\bigr).
\end{equation}
Employing a standard $L^2$ energy estimate, we obtain  the energy inequality
\begin{equation} \label{globalP10}
-\del{t}\nnorm{\del{1}^\ell \Wtt}_0^2 \leq \norm{\textrm{Div}\Att}_{L^\infty}\norm{\del{1}^\ell \Wtt}_{L^2}^2 + 
2\norm{\del{1}^\ell \Wtt}_{L^2}\norm{\Ftt_\ell }_{L^2}, \quad \ell =0,1,\dots,k, 
\end{equation}
from \eqref{globalP7},
where
\begin{equation*}
\textrm{Div}\Att = \del{t}\Att^0 + \del{1}\Att^1
\end{equation*}
and
\begin{equation*}
\nnorm{(\cdot)}^2_0 = \ip{(\cdot)}{A^0(\cdot)}
\end{equation*}
is the energy norm. 

To proceed, we define the higher energy norms
\begin{equation*}
\nnorm{\Wtt}_k^2= \sum_{\ell=0}^k\nnorm{\del{1}^\ell \Wtt}_0^2,
\end{equation*}
and observe via \eqref{globalP2} that the equivalence of norms 
\begin{equation} \label{globalP11}
 \frac{1}{\sqrt{\gamma}}\norm{\Wtt}_{H^k}\leq \nnorm{\Wtt}_k \leq \sqrt{\gamma}\norm{\Wtt}_{H^k}
\end{equation} 
holds.
Using this equivalence, we obtain, after summing \eqref{globalP10} over $\ell$ from $0$ to $k$, the differential
energy estimate
\begin{equation} \label{globalP12}
-\del{t}\nnorm{\Wtt}_k^2 \lesssim \norm{\textrm{Div}\Att}_{L^\infty}\nnorm{\Wtt}_{k}^2 + 
\nnorm{\Wtt}_{k}\left(\sum_{\ell=0}^k\norm{\Ftt_\ell }_{L^2}\right).
\end{equation}
Since $k>1/2+1$, we have by Sobolev's inequality \cite[Ch.~13, Prop 2.4]{TaylorIII:1996} that
\begin{equation} \label{globalP20}
\norm{\Wtt}_{L^\infty}+ \norm{\del{1} \Wtt}_{L^\infty} \leq C_{\text{Sob}} \norm{\Wtt}_{H^k}
\end{equation}
for some constant $C_{\text{Sob}}>0$ independent of the solution $\Wtt$. We then set
$\Rc = \frac{R}{\sqrt{\gamma}C_{\text{Sob}}}$ so that
\begin{equation} \label{globalP21}
\nnorm{\Wtt}_{k} < \Rc \quad \Longrightarrow \quad  \norm{\Wtt}_{L^\infty}  < R
\end{equation}
by \eqref{globalP11} and \eqref{globalP20}.
We also choose $\delta$, see \eqref{globalP3} above, so that
$0< \delta < \frac{\Rc}{4\sqrt{\gamma}}$ in order to guarantee that $\nnorm{\Wtt(1)}< \frac{\Rc}{4}$,
and we let $T_* \in (T^*,1),$ be the first time such that 
$\norm{\Wtt(T_*)}_{H^k}= \frac{\Rc}{2}$
or if that time does not exist, then we set $T^*=T_*$, the maximal time of existence. In either case, we have that
\begin{equation}\label{globalP22}
\nnorm{\Wtt(t)}_{k} < \frac{\Rc}{2}, \quad 0<T^*\leq T_* < t \leq 1.
\end{equation}

From the formulas \eqref{Att0def}-\eqref{Fttdef} and the bounds \eqref{Hombounds}  obeyed by $u(t)$, it is then
clear that there exists a constant $C_\ell>0$, $\ell\in \Zbb_{\geq 0}$, such that $\Att^0$, $\Att^1$ and $\Ftt$ are bounded by
\begin{gather} 
|\del{\wtt}^\ell \Att^0(t,\wtt)| + |\del{\wtt}^\ell \Att^1(t,\wtt)| \leq C_\ell , \quad
|\del{\wtt}^\ell\del{t} \Att^0(t,\wtt)|  \leq C_\ell (1+t^{2\mu-1}) \label{globalP22a}
\intertext{and}
\bigl|\del{\wtt}^\ell \Ftt(t,\wtt)\bigr|  \leq C_\ell(1+t^{2\mu-1})|\wtt| \label{globalP22c}
\end{gather}
for all $(t,\wtt)\in (0,1]\times [-R,R]$. These bounds in conjunction with the Moser and commutator estimates, see Propositions 3.7 and 3.9 from \cite[Ch.~13]{TaylorIII:1996},
and the inequality \eqref{globalP20}
imply  that
\begin{gather*}
\norm{\textrm{Div}\Att}_{L^\infty} \leq (1+t^{2\mu-1})C(\norm{\Wtt}_{H^k}) \AND
\left(\sum_{\ell=0}^k\norm{\Ftt_\ell }_{L^2}\right) \leq (1+ t^{2\mu-1})C(\norm{\Wtt}_{H^k})\norm{\Wtt}_{H^k}.
\end{gather*}
With the help of these inequalities  and  \eqref{globalP11}, we see that \eqref{globalP12} implies the energy estimate
\begin{equation*}
-\del{t}\nnorm{\Wtt}_k \leq  (1+ t^{2\mu-1})C(\nnorm{\Wtt}_{k})\nnorm{\Wtt}_{k}, \quad 0<T^*\leq T_* < t \leq 1.
\end{equation*}
By Gronwall's inequality, we obtain the bound
\begin{equation} \label{globalP24}
\nnorm{\Wtt(t)}_k\leq e^{C(\Rc)\int_{t}^1 1+\tau^{2\mu-1}\, d\tau}\nnorm{\Wtt(1)}_k, \quad 0<T^*\leq T_* < t \leq 1,
\end{equation}
where in deriving this we have used  \eqref{globalP22}.
But  
\begin{equation*}
\int_{t}^1 1+\tau^{2\mu-1}\, d\tau\lesssim 1, \quad 0<t\leq 1,
\end{equation*}
since $\mu>0$, and consequently, we have by  \eqref{globalP3}, \eqref{globalP11} and \eqref{globalP24} that
\begin{equation*}
\nnorm{\Wtt(t)}_k\leq C(\Rc)\delta, \quad 0<T^*\leq T_* < t \leq 1.
\end{equation*}
By shrinking $\delta>0$ more if necessary, it follows that
\begin{equation} \label{globalP25}
\nnorm{\Wtt(t)}_k\leq \frac{\Rc}{2}, \quad 0<T^*\leq T_* < t \leq 1.
\end{equation}
We therefore conclude by the continuation principle and the definition of $T_*$  that $T_*=T^*=0$, which establishes the existence of a unique solution 
\begin{equation*} 
\Wtt=(\ztt, \wtt )^{\tr} \in C^0\bigl((0,1],H^{k}(\Tbb^1.\Rbb^2)\bigr)\cap C^1 \bigl((0,1],H^{k-1}(\Tbb^1,\Rbb^2)\bigr)
\end{equation*}
to the initial value problem \eqref{globalS1}-\eqref{globalS2}.

Next, by integrating $\del{t}\Wtt$ in time, we  get
\begin{equation} \label{globalP26}
\Wtt(t_2)-\Wtt(t_1) = \int_{t_1}^{t_2} \del{t}\Wtt(\tau)\,d\tau, \quad 0<t_1<t_2\leq 1.
\end{equation}
Using \eqref{globalS1} to write $\del{t}\Wtt$ as
\begin{equation*}
\del{t}\Wtt= (\Att^0)^{-1}[\Att^1 \del{1}\Wtt +\Ftt], 
\end{equation*}
it is not difficult to verify from the bounds \eqref{globalP2}, \eqref{globalP11}, \eqref{globalP20}, \eqref{globalP22a}-\eqref{globalP22c}, and \eqref{globalP25}, where $T_*=T^*=0$, and the Moser estimates that
\begin{equation*}
\norm{\del{t}\Wtt}_{H^{k-1}} \lesssim 1+t^{\mu-1}.
\end{equation*} 
From this estimate and the triangle inequality, we see, after applying the $H^{k-1}$ norm to \eqref{globalP26}, that
\begin{equation} \label{globalP27}
\norm{\Wtt(t_2)-\Wtt(t_1)}_{H^{k-1}} \leq \int_{t_1}^{t_2} \norm{\del{t}\Wtt(\tau)}_{H^{k-1}}\,d\tau  \lesssim |t_2-t_1| +|t_2^{2\mu}-t_1^{2\mu}|,   \quad 0<t_1<t_2\leq 1.
\end{equation}
From this inequality, we conclude that the limit $\lim_{t\searrow 0}\Wtt(t)$, denoted $(\ztt_*,\wtt_*)$, exists in $H^{k-1}(\Tbb^1,\Rbb^2)$.
Furthermore, sending $t_1\searrow 0$ in \eqref{globalP27} shows that
\begin{equation*}
\norm{\ztt(t)-\ztt_*}_{H^{k-1}}+\norm{\wtt(t)-\wtt_*} \lesssim \norm{\Wtt(t)-\Wtt(0)}_{H^{k-1}} \lesssim t+t^{2\mu}, \quad 0\leq t \leq 1.
\end{equation*}
To complete the proof, we observe, by construction, that $\Wtt=(\ztt,\wtt)^{\tr}$ will determine a solution of the relativistic Euler equations \eqref{relEulA}
on the spacetime region $M=(0,1]\times\Tbb^3$ 
by replacing $(\zetat,w_1,w_2,w_3)$ in the formulas \eqref{relEulsol.1}-\eqref{relEulsol.5} with \eqref{symvars}.
\end{proof}

\bigskip

\noindent \textit{Acknowledgements:}
This work was partially supported by the Australian Research Council grant DP170100630.

\bibliographystyle{amsplain}
\bibliography{Kgtot_v7}

\end{document}